%% file: main.tex
\documentclass{llncs}
\usepackage{graphicx} 
\usepackage{comment}
\usepackage{amssymb}
\usepackage{amsmath}
\usepackage[ruled,vlined,linesnumbered,noresetcount]{algorithm2e}
\usepackage{caption}
\title{SmartShards: Churn-Tolerant Continuously Available Distributed Ledger}
\author{Joseph Oglio \and Mikhail Nesterenko \and Gokarna Sharma} 
\institute{Department of Computer Science, Kent State University, Kent, OH 44242, USA\\
\email{\{joglio@,mikhail@cs.,sharma@cs.\}@kent.edu}}

\begin{document}
\maketitle
\thispagestyle{plain}
\begin{abstract}

We present \emph{SmartShards}: a new sharding algorithm for improving Byzantine tolerance and churn resistance in blockchains. Our algorithm places a peer in multiple shards to create an overlap. This simplifies cross-shard communication and shard membership management. We describe \emph{SmartShards}, prove it correct and evaluate its performance.

We  propose several \emph{SmartShards} extensions: defense against a slowly adaptive adversary, combining transactions into blocks, fortification against the join/leave attack. 
\end{abstract}

\section{Introduction}
Blockchain is a distributed digital ledger maintained by a network of independent peers. The technology provides transparency and immutability of records while promises decentralized control over the ledger. This facilitates cooperation among non-trusting entities: the peers agree on the records of the ledger. The peers usually do not belong to the same organization and may join and leave the network during its operation generating continuous churn. The peers themselves may not necessarily operate correctly. Instead, they may fail or even attack the network. Such behavior is modeled as Byzantine fault~\cite{lamport1982byzantine} where the faulty peer is allowed to behave arbitrarily. 

Besides cryptocurrency~\cite{bitcoin,ethereum}, which generates a lot of recent public discourse, blockchain is used in a variety of applications: online auctions~\cite{lafourcade2020security}, marketplaces~\cite{white2022characterizing}, supply chain~\cite{dujak2019blockchain}, health care~\cite{agbo2019blockchain}, Internet-of-Things~\cite{samaniego2016blockchain}, intellectual property rights~\cite{ito2019critical}, electric power industry~\cite{andoni2019blockchain} among many others.

Popular blockchains use proof-of-work based consensus algorithms~\cite{bitcoin}
in which peers compete for the right to publish records on the blockchain by searching for solutions to cryptographic challenges. Such algorithms tend to be conceptually simple and robust. However, they are resource intensive and environmentally harmful~\cite{wendl2023environmental}. Therefore, modern blockchain designs often focus on cooperative consensus algorithms. In these cooperative consensus algorithms, 
rather than compete, peers exchange messages to arrive at a joint decision.

Within many blockchain application domains, the size of the blockchain network needs to reach thousands and possibly hundreds of thousands of peers. This necessitates systems which are highly scalable and resilient. However, traditional blockchain systems face significant scalability issues: to agree on a ledger record all peers have to communicate with each other. Hence, the communication cost grows with the network size. 

Sharding is a technique that may potentially enhance blockchain scalability. In a sharded network, the peers are divided into groups called shards. Each shard processes transactions in parallel, improving overall system throughput. A number of sharded blockchain designs are presented in the literature.
Refer to recent surveys~\cite{hafid2020scaling,liu2023survey,yu2020survey} for an extensive review. Although promising, the scalability of shading blockchain design is hampered by the inherent features of a decentralized peer-to-peer systems. In particular, cross-shard transactions and churn. 

In this paper we propose an approach of shard construction we call \emph{SmartShards} that addresses both of these problems. \emph{SmartShards}  overlaps shards to improve shard communication coordination and membership. 

\ \\
\textbf{Cross-shard transactions and faults.} A shard acts as a unit of either coordination or data ownership. For example, in Elastico~\cite{luu2016secureelastico}, the shards concurrently verify different transactions.
In \emph{Rapidchain}~\cite{zamani2018rapidchain}, each shard stores data about a collection of wallets. A cross-shard transaction requires the coordination of multiple shards. For example, in \emph{Rapidchain}, a transaction may be moving funds between wallets in two different shards. There are multiple ways to coordinate such transactions. Some blockchains use locking~\cite{kokoris2018omniledger}. Other designs break the single transaction into sub-transactions to be executed separately. Alternatively, the system may have a dedicated reference committee that approves
cross-shard transactions~\cite{luu2016secureelastico}. Regardless of the technique, the source and target shards have to communicate either directly or though intermediate coordinators.

Individual peers may be faulty. One approach to Byzantine fault-tolerance is to use cryptographic signatures. However, this may be computationally expensive or difficult to implement in a peer-to-peer system. Instead, we consider a classic approach where the correctness depends on sufficiently large majority of correct peers. 
To ensure information propagation correctness, multiple source shard peers have to broadcast the transaction data to multiple target shard peers. This limits the efficiency of cross-shard transaction processing. In \emph{SmartShards} that we propose, each peer participates in multiple shards. The shard overlap is used to make inter-shard communication more efficient and robust. 

\ \\
\textbf{Churn.}
Churn, continuous joining and leaving of peers, may disrupt the operation of a blockchain. In case of sharding, the problem is exacerbated since it may potentially lead to
shard failure and structural network compromise. However, the problem of handling continuous churn within sharded blockchains remains underexplored. 

In most systems~\cite{kokoris2018omniledger,luu2016secureelastico,zamani2018rapidchain}, the churn is handled through periodic system re-configuration events. An epoch is the time of continuous system operation between such reconfiguration events. For such epoch-based systems, the rate of churn is considered low enough to ensure adequate system availability. If the churn is substantial, to ensure correct blockchain operation, the system has to spend significant time re-configuring itself. This downtime increases with both the rate of churn and the scale of the system. 

In a sharded blockchain, peers have to maintain both their shard membership and the links to the other shards for cross-shard transactions. In traditional systems, this information has to be maintained separately. In \emph{SmartShards}, the peers are used for internal shard recording and for inter-shard communication. Therefore, the overlapping shards integrally maintain both membership and links. This simplifies churn handling in \emph{SmartShards}. 



\ \\
\noindent
\textbf{Related work.} 
There are several blockchain designs that overlap coordination between consensus groups. \emph{Monoxide}~\cite{wang2019monoxide} is a competitive Proof-of-Work blockchain. 
To improve scalability, it uses multiple asynchronous consensus zones that maintain separate blockchains. These blockchains periodically synchronize by linking blocks across 
zones. A peer may work on any of these local blockchains. Thus, \emph{Monoxide} has fixed zones with floating peer membership. \emph{Hyperledger Fabric}~\cite{androulaki2018hyperledger} uses cooperative consensus. It also maintains multiple blockchains. It has a centralized ordering service for transaction synchronization. Peers may join and leave the network only when the system is off-line, i.e. outside an epoch. 

Churn management in distributed systems has been extensively studied.
Kuhn {\em et al.}~\cite{kuhn2010towards} propose simulating a robust node by a collection of peers connected to peers of other simulated nodes. These simulated nodes retain there data despite individual peer churn. In this sense, these simulated nodes are similar to shards in blockchain networks. 
Foreback {\em et al.}~\cite{foreback2016infinite,foreback2019churn} explore the theoretical limits of infinite churn, proposing techniques to maintain network connectivity despite continuous churn in general peer-to-peer networks.

Malkhi and Reiter~\cite{malkhi1998byzantine} study Byzatine-robust quorum systems for distributed storage. A quorum is a collection of subsets of peers where each pair of subsets intersects. This arrangement is similar to \emph{SmartShards}. However, in their approach, quorums are used to 
build a robust register and to synchronize the entire network and ensure consistent reading and writing of the shared data. In \emph{SmartShards}, the quorums are used to maintain separate data in separate shards and only synchronize if the data is affected by smart-shard transactions.

\emph{SMARTCHAIN}~\cite{bessani2020byzantine} is a blockchain algorithm that uses quorums to improve performance. However, \emph{SMARTCHAIN} uses cryptographic signatures to verify cross-shard transactions. In contrast, \emph{SmartShards} does not need the cryptographic signatures. 

Baldoni {\em et al.}~\cite{baldoni2013protocol} use quorums for churn resistance in Byzantine fault-tolerant distributed storage. Their approach requires bounded churn and eventual system synchrony. \emph{SmartShards} do not need such assumptions.

García-Pérez {\em et al.}~\cite{garcia2019federated} use overlapping quorums to achieve consensus despite Byzantine faults. Similar to \emph{SmartShards}, the quorum intersections are used for data update synchronization. However, a ``trust relationship'' is required for quorum formation and maintenance. \emph{SmartShards} do not need such trust relationship. 

\ \\
\textbf{Our contribution.}
To summarize, \emph{SmartShards} presented in this paper is a novel sharded approach to Byzantine tolerance and churn resistance in blockchain systems. 

\section{Notation and Problem Definition}

\textbf{System model.} We assume a peer-to-peer network. Each peer has a unique identifier. A peer may send a message to any
other peer so long as it has the receiver's identifier. Communication is reliable.
Peers communicate through authenticated channels: the receiver of the message may always identify the sender. 

\begin{figure}
\captionsetup{width=0.47\textwidth}
\begin{tabular}{c c}
    \centering
    \begin{minipage}[t]{0.45\textwidth}
        \centering
        \includegraphics[width=0.60\textwidth]{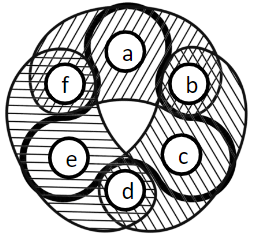}
        \caption{\emph{SmartShards} peer to shard allocation example. The shards are: $\{f, a , b\}$, $\{b,c,d\}$, $\{d, e, f\}$,  $\{a, c, e\}$.}
        \label{smartshards-diagram}
    \end{minipage}
    &
    \begin{minipage}[t]{0.55\textwidth}
        \centering
        \includegraphics[width=0.7\textwidth]{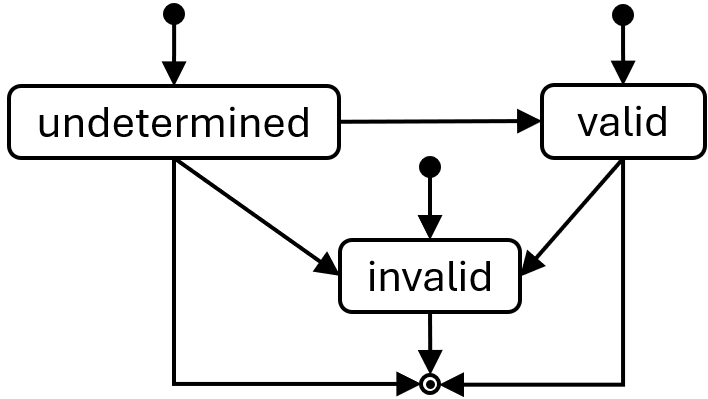}
        \caption{Transaction validity transitions.}
        \label{validity-diagram}
    \end{minipage}
\end{tabular}
\end{figure}

\noindent
\textbf{Shards, transactions, faults.}
A \emph{shard} is a fully connected set of peers maintaining a blockchain. 
Shards $A$ and $B$ are considered \emph{overlapping} if there exists at least one peer $p \in A \cap B$. For simplicity, 
we initially assume that a peer may belong to at most two shards. See example peer to shard allocation in Figure~\ref{smartshards-diagram}.

The overlaps are used for communication between shards. We assume that the overlap size, $x$, is the same for all shards. 

Two peers $p$ and $q$ that belong to the same shard are \emph{shard mates}, or just \emph{mates}. If $q$ is a mate of $p$ in some shard, then the \emph{countershard} of $p$ for $q$  is the other shard to which $q$ belongs. 

A \emph{wallet} is a means of storing funds. Each wallet has a unique identifier. A \emph{client} is an entity that \emph{owns} wallets. This ownership can be authenticated by the peers of the network. There may be multiple clients. A client submits transactions. A \emph{transaction} is a transfer of funds from a \emph{source} to a \emph{target} wallet. 
For simplicity, we assume that each transaction has a single source and a single target wallet. A shard records transactions for a set of wallets.
A \emph{source shard} records the transaction if the source wallet is in this shard's set. Similarly, a \emph{target shard} records a transaction if it maintains the target wallet. 
A client may have wallets in more than one shard. 

There are two types of transactions. A \emph{cross-shard transaction} moves funds between wallets of different shards.
An \emph{internal transaction} moves funds between two wallets of the same shard. Note that for an internal transaction, the single shard is both the source and the target. 

The funds are transferred in the form of \emph{UTXOs} (unspent transaction outputs). A transaction is applied to an \emph{input} UTXO and produces an \emph{output} UTXO.

A ledger is a sequence of transactions recorded by a shard. Each peer in the shard maintains its own copy of this ledger. The peers are initialized with the same genesis ledger.

A transaction is \emph{confirmed} by a  peer if it is recorded in this peer's ledger. An output UTXO of a confirmed transaction is either spent or unspent.  A UTXO is \emph{unspent} if it is an output of a confirmed transaction and there is no other confirmed transaction that uses it as an input. It is \emph{spent} if used as an input to a confirmed transaction. If a transaction is unconfirmed, its output UTXO is neither spent nor unspent. 

A transaction, whether confirmed or not, is \emph{valid} if it is applied to an unspent UTXO. A transaction is \emph{invalid} if it is applied to a spent UTXO. Note that if the input of a transaction is neither spent nor unspent, then its validity is undetermined. 
The validity transitions of a transaction are, therefore, as follows (see Figure \ref{validity-diagram}). A submitted transaction may be either valid, invalid, or undetermined. 
An undetermined transaction may eventually  become either valid or invalid. A transaction of determined validity may never become undetermined again. 
A valid transaction may become invalid if another transaction that spends its input UTXO is confirmed. A valid confirmed transaction is never invalidated. An invalid transaction, either confirmed or unconfirmed, never becomes valid.

A \emph{Byzantine} peer is faulty. A faulty peer behaves arbitrarily. A \emph{correct} peer is not faulty.

\ \\
\textbf{Consensus and Distributed Ledger Problems.}

\begin{definition}
In \emph{the Consensus Problem}, 
every correct process is input a value $v$ and must output an irrevocable decision subject to the following properties: 
\begin{description}
\item[CValidity:] if all correct peers are input the same value $v$, then every correct peer decides $v$; 
\item[CAgreement:] no two correct peers decide differently; 
\item[COrder:] any pair of decisions are ordered the same of every pair of peers; 
\item[CLiveness:] every correct peer eventually decides.
\end{description}
\end{definition}

\emph{Tolerance threshold} $f$ is the maximum number of faulty processes. 
An algorithm solves the Consensus Problem if it satisfies the three properties above provided that the number of faults does not exceed $f$.

\begin{definition} In \emph{the Distributed Ledger Problem}, peers confirm a particular transaction with the following three properties:

\begin{description}
\item[LValidity:] every confirmed transaction is valid; 
\item[LAgreement:] if a transaction is confirmed by one correct peer in a shard, it is confirmed by every correct peer in the shard; and
\item[LLiveness:] if a client submits a valid transaction, it is eventually either confirmed or invalidated.
\end{description}

\end{definition}



\begin{figure}[htb]
\fontsize{8pt}{9pt}\selectfont
\input{variables}
\end{figure}

\begin{figure}[htb]
\fontsize{8pt}{9pt}\selectfont
\input{algorithm}
\end{figure}

\begin{figure}[htb]
\fontsize{8pt}{9pt}\selectfont
\input{algorithm3}
\end{figure}

\section{\emph{SmartShards} Description and Correctness Proof}

\textbf{Algorithm outline.} \emph{SmartShards} uses a consensus algorithm \emph{Consensus} to confirm transactions. The peers of the source shard execute \emph{Consensus} to agree on the transaction. This is sufficient for an internal transaction where the source and target shards are the same. For a cross-shard transaction, the peers that overlap the source and target shard, inform the target shard of the transaction via a \textsc{Transfer}  message. Once a peer of the target shard receives sufficient number of \textsc{Transfer} messages, it initiates consensus on the cross-shard transaction. 

\ \\
\textbf{Constants and variables description.}
The constants and variables are listed in Algorithm~\ref{algVars}. Let us start with constants. 
Each peer has a unique identifier $p$. The peer stores a set of all shard ids in $Shards$. Out of this set, $p$ belongs to two shards: $shard_1$ and $shard_2$. The set $Peers_1$ contains the ids of the mates of $p$ in $shard_1$. Similarly, $Peers_2$ are mates of $p$ in $shard_2$.  
For each mate $q$ of $p$,  $CounterShards_1$ lists a set of countershards for all mates of $p$ in $shard_1$. $Countershards_2$ is defined similarly for $shard_2$. 

The variables maintained by $p$ are of two kinds: transaction and process variables. Transaction variables pertain to a particular transaction: client id, source and target wallet, source UTXO, and source and target shard. Note that if the transaction is internal, then the source and target shards are the same. 

In addition to transaction variables, $p$ maintains process variables in all transactions. Specifically, $Ledger_1$ contains records of confirmed transactions for $shard_1$. Variable $Ledger_2$ keeps the records for $shard_2$. Variable $Transfers_1$ collects notifications of confirmed transactions from the source shard peers.

\ \\
\textbf{Actions description.} \emph{SmartShard} actions are listed in Algorithm~\ref{algActions}.
The function \textsc{Valid} is used in the actions of \emph{SmartShards}. It examines the shard ledger of $p$ to check if the specific $UTXO$ is in the client wallet and is unspent. That is, this $UTXO$ is not an input to another transaction in this ledger. 

There are four actions in the algorithm: \textsc{Transaction}, \textsc{EndConsensusTx}, \textsc{EndConsensusTr}, and \textsc{Transfer}. 
The \textsc{Transaction} action, see Line~\ref{transaction}, processes client transaction receipts in the source shard. If a transaction is valid, the peer starts the consensus algorithm for this transaction.

Once this consensus is done, the \textsc{EndConsensusTx} action, see Line~\ref{endConsensusTx}, at the source shard checks if the transaction is still valid and confirms the transaction by recording it in the local ledger and sends a confirmation message to the client. If the transaction is cross-shard, and the peer is in both  source and target shards, then the peer sends a \textsc{Transfer} message to the target shard notifying it of the successful confirmation. 

The \textsc{Transfer} action handles the \textsc{Transfer} message receipt, see Line~\ref{transfer}. A recipient peer $p$ in $shard_1$ collects such receipts in $Transfers_1$.
If $p$ receives \textsc{Transfer} messages from greater than half of the peers in the overlap between the source and target shards, it considers the transaction to be valid and
then $p$ initiates consensus in the target shard on the cross-shard transaction. Every peer in the target shard participates in this consensus. Each peer uses the number of received \textsc{Transfer} messages to set its initial value for consensus. If the number of received \textsc{Transfer} messages is greater than half the number of peers in the overlap then the value is set to \textbf{true} otherwise \textbf{false}.

Once consensus on a transfer is reached, the \textsc{EndConsensusTr} action, see Line~\ref{endConsensusTr}, at the target shard checks whether the peers agree on the validity of the transaction. If agreement is reached, and $result$, the consensus value, is \textbf{true}, the transaction is recorded in the ledger and sends a confirmation message to the client.

\ \\
\textbf{Correctness proof.}

\begin{lemma}\label{lemLValidity}
\emph{SmartShards} satisfies the LValidity property of the Distributed Ledger Problem. 
\end{lemma}
\begin{proof} 

To satisfy LValidity, we need to show that only valid transactions are confirmed.
Let us consider internal transactions first. In this case, the source and target shards are the same. Therefore, \emph{SmartShards} executes only the \textsc{Transaction} and \textsc{EndConsensusTx} actions without transmitting any \textsc{Transfer} messages. The peers of the shard execute the \emph{Consensus} algorithm for the transaction within a single shard.  An internal transaction is only confirmed by the \textsc{EndConsensusTx} action.  
Before confirming, the action checks the transaction validity. That is, only valid transactions are confirmed. Valid confirmed transactions may not be invalidated. 

Let us now deal with cross-shard transactions. Such transactions are confirmed in the source and target shards separately.  By an argument similar to the above, we can show that a cross-shard transaction is confirmed in the source shard only if it is valid. If such a transaction is confirmed in the source shard, the source shard peers send a \textsc{Transfer} message. This initiates consensus in the target shard. A correct peer starts such consensus only after it receives a \textsc{Transfer} message from more than half of the peers in the overlap between the source and target shards. This means that at least one correct peer in the source shard sends such message. That is, this transaction is valid.

Note that a cross-shard transaction is confirmed by the \textsc{EndConsensusTr} action in the target shard. If this consensus is started by a correct peer, the transaction is valid. 
Alternatively, a Byzantine peer may start \emph{Consensus} on a potentially invalid transaction. For such a transaction, correct peers do not have a sufficient number of \textsc{Transfer} messages. Therefore, all correct peers are input an initial value of \textbf{false}. By the CValidity property of consensus, if all correct peers are input \textbf{false}, they must decide \textbf{false}. In which case, the transaction is not confirmed. 
Thus, \emph{SmartShards} satisfies LValidity for both internal and cross-shard transactions.
\qed
\end{proof}

\begin{lemma}\label{lemLAgreement}
\emph{SmartShards} satisfies the LAgreement property of the Distributed Ledger Problem. 
\end{lemma}
\begin{proof} 
We prove this property by induction. Each peer starts with the same genesis ledger which is assumed to contain valid transactions. Let the ledger at each peer contain $i$ valid transactions.  Any transaction is confirmed by either the execution of the \textsc{EndConsensusTx} or \textsc{EndConsensusTr} action. Let us consider the action for transaction $i+1$.   By COrder property,
the order of \textsc{EndConsensusTx} and \textsc{EndConsensusTr} actions are the same for every peer. By CAgreement, this is the same transaction at every peer. Since the ledger up to the $i$-th transaction is the same for each peer, the validity of this transaction is also going to be the same. That is, this transaction is going to be either confirmed or rejected by all correct peers. Thus, the ledger remains consistent after processing $i+1$ transactions. 
\qed
\end{proof}

\begin{lemma}\label{lemLLiveness}
\emph{SmartShards} satisfies the LLiveness property of the Distributed Ledger Problem. 
\end{lemma}
\begin{proof}
For internal transactions, the peers are running \emph{Consensus}. Therefore, LLiveness is satisfied because the CLivenss property of \emph{Consensus} ensures that every transaction is eventually either confirmed or invalidated. 

Let us discuss cross-shard transactions. Similarly to internal transactions, the peers of the source shard are running \emph{Consensus}. Therefore, they eventually either confirm or invalidate any cross-shard transaction. If the transaction is confirmed by the source shard, the source shard peers that overlap the target shard send \textsc{Transfer} messages.
Greater than half of such peers are assumed to be correct. Therefore, each target peer eventually gets more than half such \textsc{Transfer} messages. 

If a target peer gets more than half \textsc{Transfer} messages, it will start \emph{Consensus}. If all correct peers received \textsc{Transfer} messages from more than half of the overlap, then they all use \emph{true} as the input for \emph{Consensus}. By CValidity, this consensus must succeed. Therefore, LLiveness holds for cross-shard transactions as well. 
Hence the lemma. 
\qed
\end{proof}

The below theorem follows from Lemmas~\ref{lemLValidity},~\ref{lemLAgreement}, and~\ref{lemLLiveness}. Recall that $x$ is the shard overlap size.

\begin{theorem} \label{thrmSmartShards}
Algorithm \emph{SmartShards} solves the Distributed Ledger Problem with at most $f$ Byzantine faults in each shard and at most $\lfloor x/2 -1 \rfloor$ faults in the overlap between shards. 
\end{theorem}

\noindent
\textbf{Fault tolerance threshold estimation.}
Let us address the relation between $x$ and $f$. From the above theorem it follows that  $f < x/2$. However, we may be able to get a more optimistic bound. In classic sharding algorithms, it is assumed that faults are distributed across shards somewhat evenly. To put another way, the number of faults in a shard is proportional to shard size. We apply this reasoning to the shard overlap size. 

Let $s$ be the number of shards in the network and $z$ be the shard size. Assume the shard fault tolerance threshold is $f < z/3$. Since every pair of shards overlap over $x$ peers, the shard size is $z = x(s-1)$. If the number of faults is proportional to overlap size then the tolerance threshold can be estimated as follows:
\[
f < \frac{z}{3} = \frac{x(s-1)}{3}.
\]

\noindent
\textbf{Churn tolerance.} SmartShards may be modified to tolerate churn: peer joining and leaving. 
There are two aspects: (i) the client needs to ascertain that the appropriate member peers confirm the transaction and (ii) the peers themselves need to agree on the current shard membership.

We assume the existence of the Shard Membership Service (SMS) -- a trusted service providing membership information to the clients.  A client submits its transaction to the SMS. The SMS forwards this transaction to the peers in the source shard and replies to the client with the current set of peer identifiers in this shard. Once the peers receive the transaction from the SMS, they run the consensus algorithm and send the confirmation to the client. 

A client may re-send the transaction. In this case, the SMS behaves similarly: it forwards the transaction to the peers and sends the membership set to the client. Due to peer churn, the membership set may differ every time. If the peers receive a request for a transaction that is already confirmed, they re-send the confirmation to the client. 

Potentially, the SMS may return the exact set of peers that handle the transaction. However, such service may be difficult to implement as it needs to synchronize its operation with peers joining and leaving.  Instead, we consider a weaker SMS assumption. If the client submits infinitely many requests for a specific valid transaction, then the SMS provides infinitely many sets containing at least $f+1$ peers that reply to the transaction request. \emph{Weak SMS} is the SMS that satisfies this assumption. Henceforth, we refer to weak SMS as just SMS.  

The client actions are shown in Algorithm~\ref{algClient}. We show the code for a single transaction confirmation. Until confirmed, the client periodically sends the transaction requests to the SMS. The client collects membership sets sent back by the SMS in $SMSmembers$. The client also collects peer confirmations in $Confirmations$. The client considers the transaction confirmed if it gets at least $f+1$ confirmations for this transaction from one of the membership sets received from the SMS. 

Let us now discuss churn. A peer wishing to join or leave \emph{SmartShards} submits a transaction with this request. A joining peer utilizes the SMS. The process is similar to a client submitting a regular UTXO-related transaction. With their confirmation, the peers in the shard send their complete ledger to the joining peer. Once this peer receives $f+1$ confirmations from one of the shard membership sets supplied by the SMS, the peer is ready to participate in \emph{SmartShards}. The algorithm for a leaving peer is simpler since it knows the shard membership. The leaving peer just submits a leave request to the members of its shard. It may leave the shard once it receives $f+1$ confirmations. After the peer leaves, it may not receive or send any messages. 
We summarize this discussion in the below theorem.

\begin{theorem} \label{thrmSmartShardswChurn}
Algorithm \emph{SmartShards} with weak SMS solves the Distributed Ledger Problem with at most $f$ Byzantine faults in each shard and at most $\lfloor x/2 -1 \rfloor$ faults in the overlap between shards despite peer churn.
\end{theorem}

\begin{figure}
\captionsetup{width=0.47\textwidth}
\begin{tabular}{c c}
    \centering
    \begin{minipage}[t]{0.5\textwidth}
        \includegraphics[width=1\textwidth]{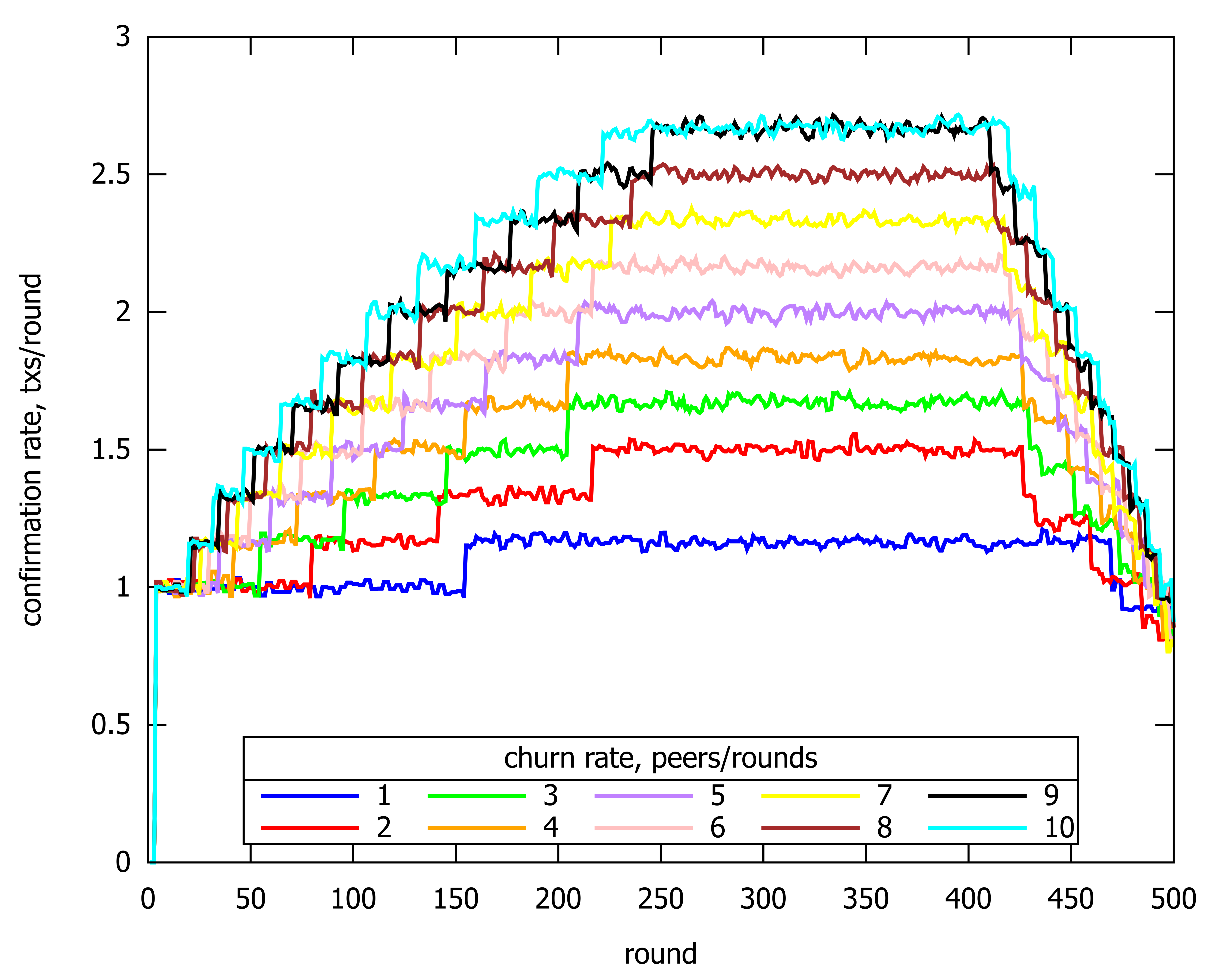}
        \caption{Transaction confirmation rate timing diagram. Peers request to join the network every round from $1$ to $250$. They request to leave the network every round from $251$ to $500$.}
        \label{timing-diagram}
    \end{minipage}
    &
    \begin{minipage}[t]{0.5\textwidth}
        \centering
        \includegraphics[width=1\textwidth]{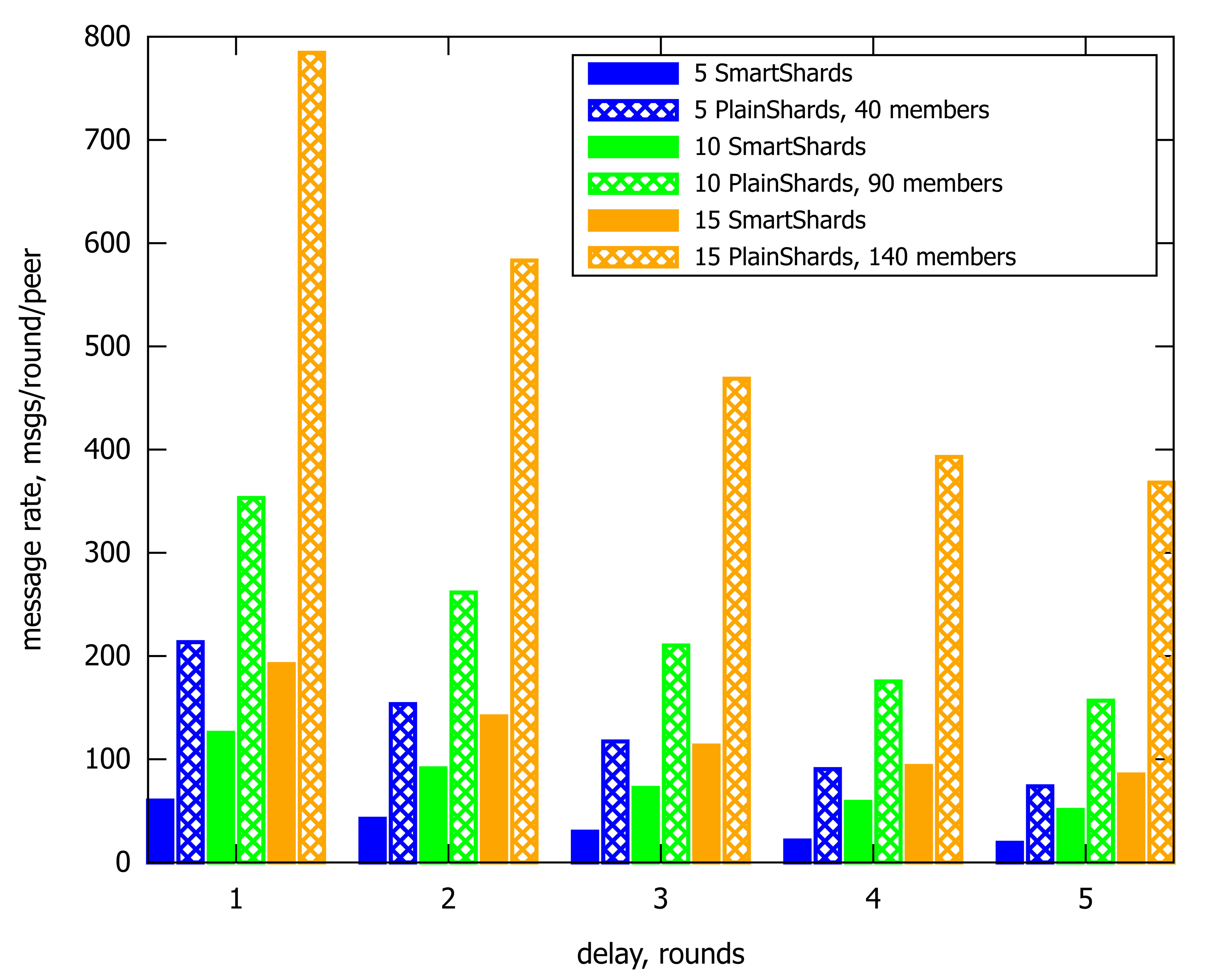}
        \caption{Message rates for varied shard sizes and max message delay for the \emph{SmartShards} algorithm vs. \emph{PlainShards}.}
        \label{smart-vs-notsmart}
    \end{minipage}
\end{tabular}
\end{figure}

\section{Performance Evaluation}

\noindent\textbf{Simulation setup.}
We evaluate the performance of \emph{SmartShards} in the abstract algorithm simulator QUANTAS~\cite{quantas}. The QUANTAS code for the \emph{SmartShards} simulation as well as our performance evaluation data is available online~\cite{SmartShardsgithub,smartShardsData}.

The simulated network consists of individual peers. Each pair of peers communicate via a message-passing channel. The channels are FIFO and reliable. A single computation is modeled as a sequence of rounds. In each round, a peer receives messages that were sent to it in the previous rounds, performs local computation, and sends messages to the other peers. 

The receipt of each message may be delayed up to some maximum amount which various by experiment. The network is modeled with  $n$ peers divided into $s$ shards.

We do not explicitly simulate clients or faults. Each shard processes transactions independently. We evaluate \emph{SmartShards} where peers participate in multiple overlapping shards against \emph{PlainShards} where each peer is in at most one shard. In \emph{SmartShards} each peer is in at most two shards. If two shards overlap, $x$ is the number of peers in this overlap. Therefore, total number of peers in the network is subject to this formula:
\[n = \frac{s(s-1)x}{2}\]
In \emph{SmartShards}, each shard is running \emph{PBFT}.
\emph{PBFT}~\cite{pbft} is a common leader-based consensus algorithm. Transactions are proposed by the \emph{PBFT} leader of each shard.  All transactions are valid and have a chance of being cross-shard that is governed by the number of shards in the network. 

Peers in the network may join and leave shards. Unless otherwise stated, in each round, a single peer requests to join and a single peer requests to leave the network. 

\emph{SmartShards} and \emph{PlainShards} differ in their implementation of cross-shard transactions and transaction membership. Since there is no shard overlap, in case of a cross-shard transaction, each source shard peer informs every target peer of source shard consensus decision. Also, each peer has to maintain memberships of all network shards. Therefore, once a peer joins or leaves the network, it has to inform all network peers.

We simulate computations of $500$ rounds. Unless otherwise stated, each data point is an average of $10$ tests.

\ \\
\noindent\textbf{Experiment description.} In the first experiment we observed the dynamics of \emph{SmartShards} churn handling. The results of a representative computations are shown in Figure~\ref{timing-diagram}. In every round from $1$ to $250$, new peers request to join the network. Then,  in rounds $251$ though $500$, randomly selected peers request to leave the network. \emph{PBFT} leaders and peers that have not joined yet, are not selected to leave. If there are least $25$ peers in the shard overlap between each shard on average, a new shard is created, random peers from donor shards join a new shard.  If there are less than $9$ peers in the shard overlap between each shard on average, a shard is destroyed and its peers are randomly distributed to other shards. We plot a rolling average of transaction confirmation rate over $100$ tests with a window of $3$ rounds.  We vary the churn rate: the number of peers requesting to join or leave the network every round. 

Initially, the network contains $100$ peers divided into $5$ shards with an overlap of $10$. The confirmation rate reflects the number of available shards: the greater the number of shards the higher the confirmation rate. The experiments show the robustness of \emph{SmartShards} with the respect to churn.

Figure~\ref{smart-vs-notsmart} compares \emph{SmartShards} with \emph{PlainShards}, the non-overlapping sharded implementation. We compare the message exchange rate for \emph{SmartShards} and \emph{PlainShards} networks of comparable sizes. The message rate is the total number of messages exchanged  divided by the number of rounds in the computation and the number of peers in the network. The figure indicates that \emph{SmartShards} outperforms \emph{PlainShards} due to more efficient implementation of cross-shard transaction and churn handling. Each \emph{SmartShards} experiment uses an overlap of $10$.

For the next set of experiments, we vary message delay and observe its influence on network performance. In Figures~\ref{shard-sizes-throughput},~\ref{shard-sizes-latency}, and~\ref{shard-sizes-messages}, we change the number of shards in \emph{SmartShards} and observe the dynamics of throughput, latency and message rate respectively. Throughput is computed as the total number of transactions confirmed during the computation. Latency is the number of rounds that elapse from when a transaction is submitted until it is confirmed by \emph{PBFT}. The figures indicate that with the increase of the number of shards the performance of \emph{SmartShards} improves with a relatively modest increase in message expense.  

In Figures~\ref{intersection-sizes-throughput},~\ref{intersection-sizes-latency}, and~\ref{intersection-sizes-messages} we change the intersection size and observe its influence on throughput, latency and message rate. Figures~\ref{intersection-sizes-throughput} and~\ref{intersection-sizes-latency} demonstrate that as the number of intersections between shards in the network increases, the throughput  and latency remain unaffected. However, Figure~\ref{intersection-sizes-messages} shows that the increase in the number of intersections increases the message rate as expected due to the increase in shard sizes. This cost is relatively low due to the benefit of increasing the Byzantine tolerance threshold of the algorithm.

\begin{figure}
    \captionsetup{width=0.49\textwidth}
\begin{tabular}{c c}
    \begin{minipage}[t]{0.5\textwidth}
        \includegraphics[width=\textwidth]{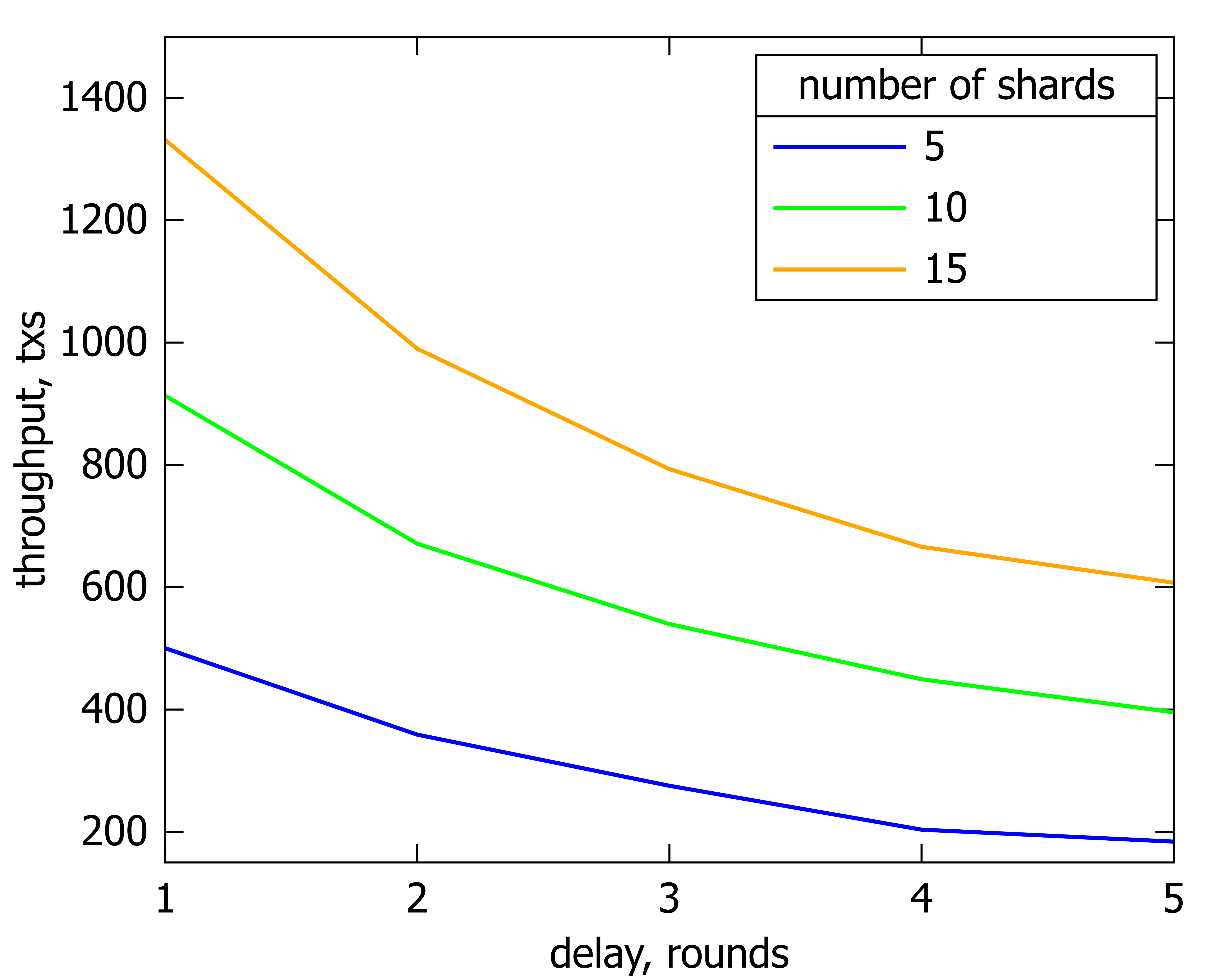}
        \caption{Number of approved transactions depending on the message delay for various number of shards.}
        \label{shard-sizes-throughput}
\end{minipage}
&
\begin{minipage}[t]{0.5\textwidth}
        \includegraphics[width=\textwidth]{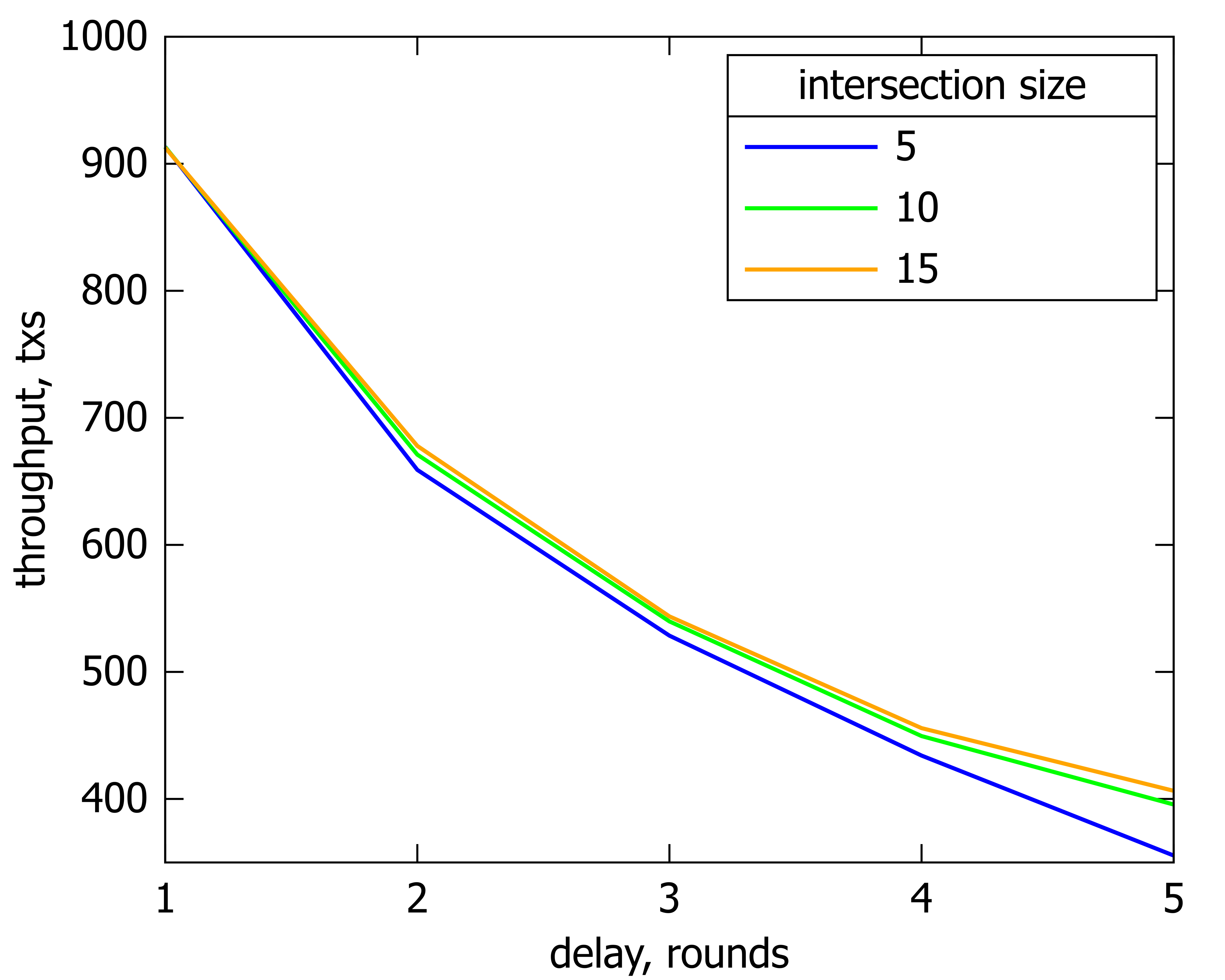}
        \caption{Number of approved transactions depending on the message delay for various intersection sizes.}
        \label{intersection-sizes-throughput}
\end{minipage}

\\

\begin{minipage}[t]{0.5\textwidth}
        \includegraphics[width=\textwidth]{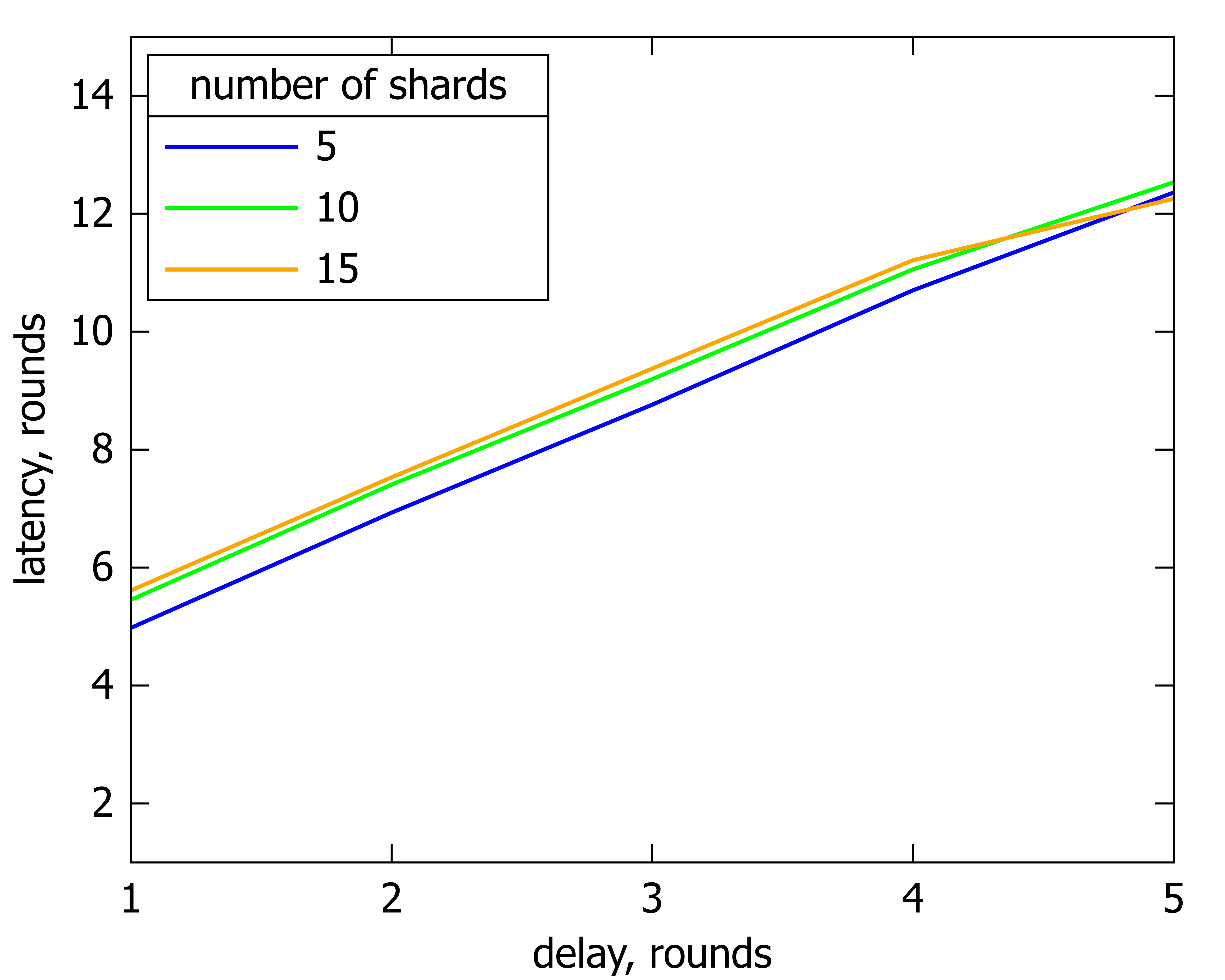}
        \caption{Latency of approving a transaction depending on the message delay for various number of shards.}
        \label{shard-sizes-latency}
\end{minipage}
 &
\begin{minipage}[t]{0.5\textwidth}
       \includegraphics[width=\textwidth]{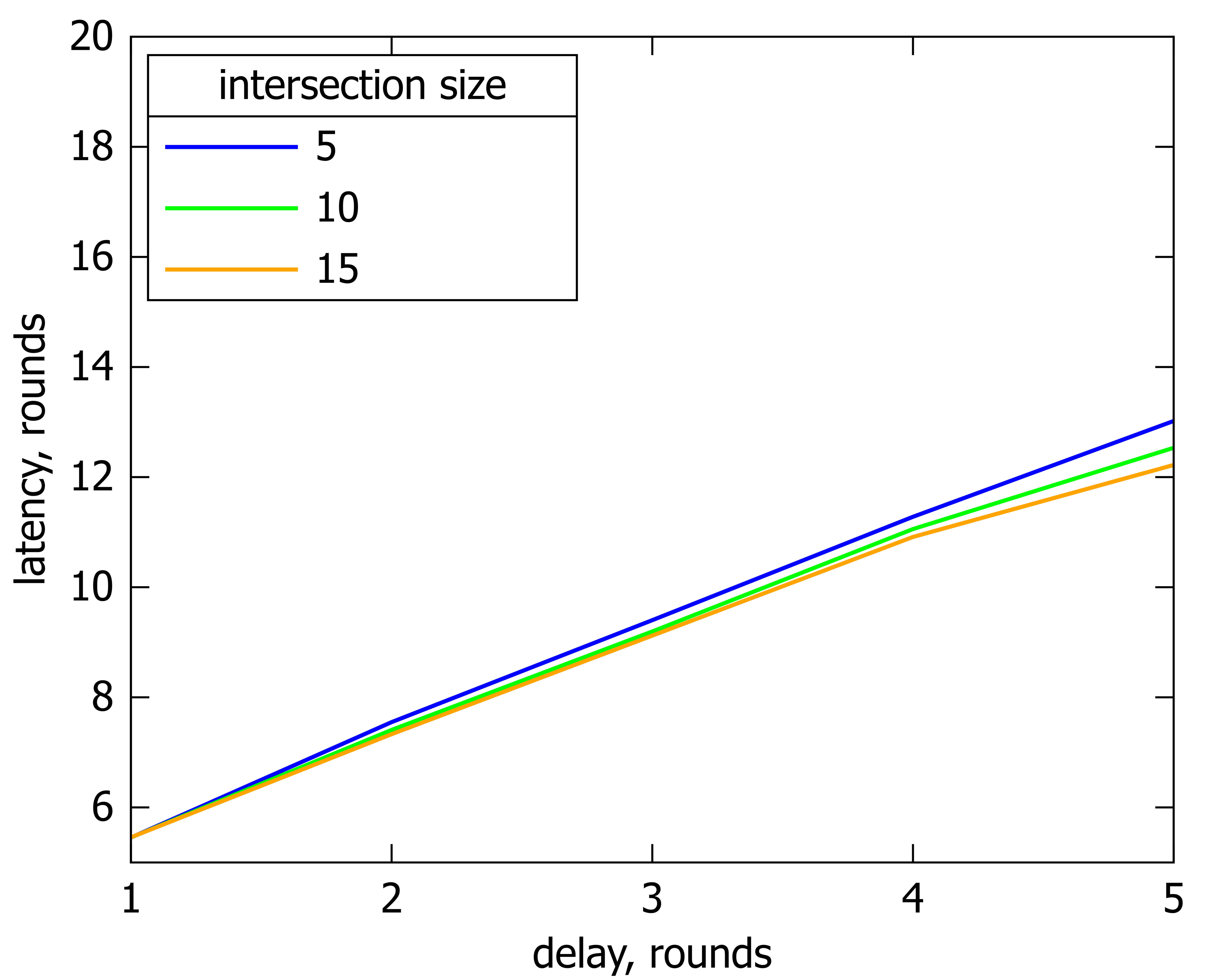}
        \caption{Latency of approving a transaction depending on the message delay for various intersection sizes.}
        \label{intersection-sizes-latency}
\end{minipage}

\\

\begin{minipage}[t]{0.5\textwidth}
        \includegraphics[width=\textwidth]{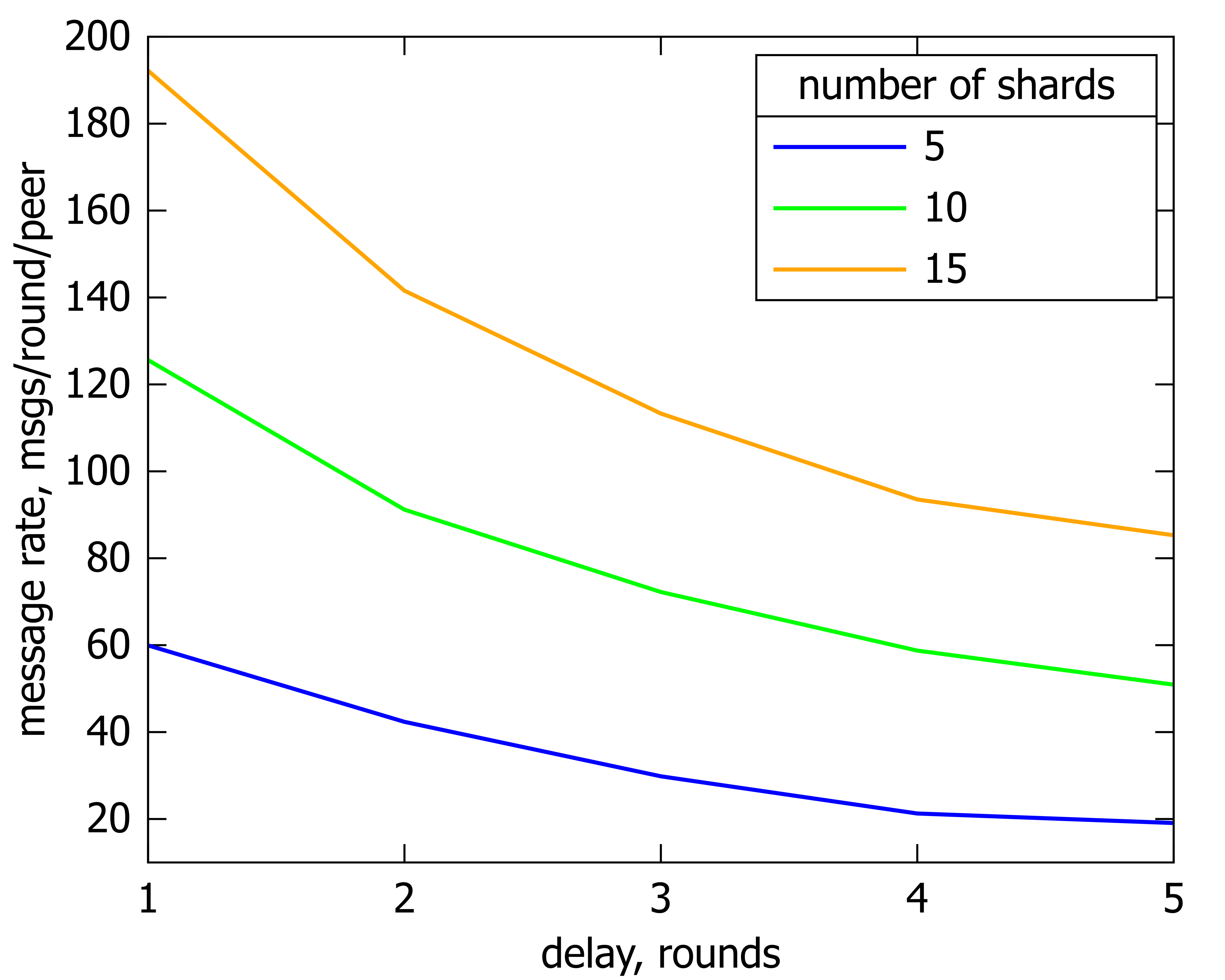}
        \caption{Message exchange rate depending on message delay.}
        \label{shard-sizes-messages}
\end{minipage}
&
\begin{minipage}[t]{0.5\textwidth}
       \includegraphics[width=\textwidth]{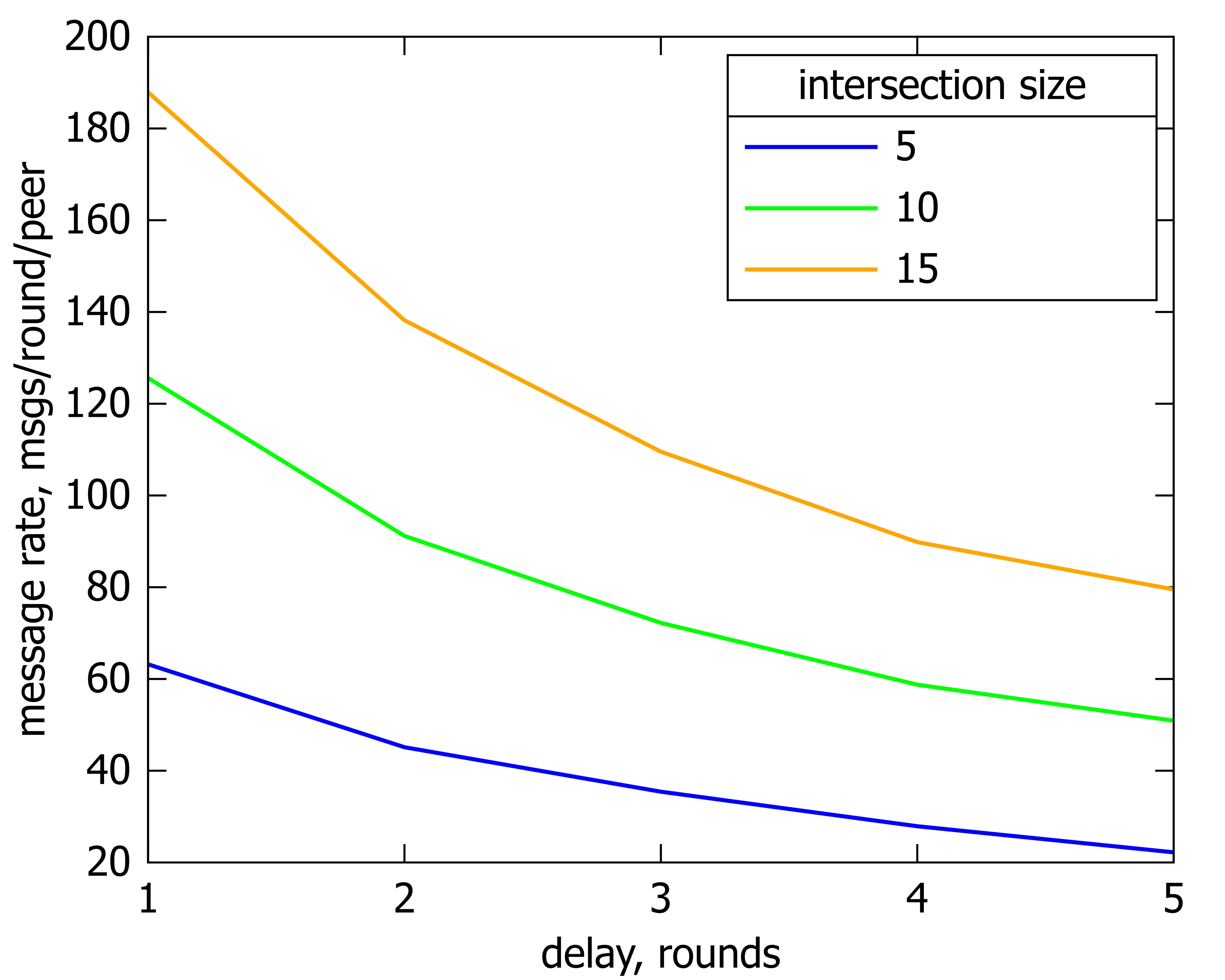}
        \caption{Average message rate for varied intersection sizes and max message delay.}
        \label{intersection-sizes-messages}
\end{minipage}
    
 \end{tabular}
 \vspace{-8mm}
 \end{figure}

\vspace{-.4cm}
\section{Extensions and Implementation Considerations}

\vspace{-.2cm}
\textbf{Multiple sources and targets transactions.} \emph{SmartShards} can be adapted to handle a \emph{multi-transaction} that transfers funds from multiple source wallets or to multiple target wallets atomically. If the wallets belong to the same shard, the modification is simple since the peers of the same shard confirm the transaction. 
 However, it becomes complicated if the wallets are located in different shards.

In this case, the multi-transaction is preceded by the Consolidation Phase and followed by the Distribution Phase. The actual transaction is carried out in a \emph{consolidation shard}. This shard may be one of the source or target shards or an unrelated shard.
The \emph{Consolidation Phase} consists of moving the source UTXOs into the consolidation shard. Each source UTXO is transferred to the consolidation shard without changing ownership of the UTXO. That is, the same client owns the wallet which contained the UTXO in the source as in the consolidation shard. After the Consolidation Phase is complete, \emph{SmartShards} confirms the multi-transaction in a single consolidation shard. After this execution, the target UTXOs are transferred to the target shards during the \emph{Distribution Phase}. Again, there is no ownership change during such a transfer. 

\ \\
\textbf{Defense against membership attacks.} \emph{SmartShards} may be fortified against a \emph{slowly adaptive adversary}~\cite{zamani2018rapidchain} that requires a certain time period to corrupt new peers. Such an adversary may focus on compromising a single shard or a single shard intersection. If the corruption period is known, peers are ejected from their shards before a timer expires to prevent a Byzantine majority from compromising the system. Such peers then join a different shard shuffling the peers and distributing the faults.

Similarly, defense against a leave/join attack may be incorporated into \emph{SmartShards}. In such an attack, the faulty nodes repeatedly join and leave the network hoping to get into the same shard in an attempt to exceed its tolerance threshold. To counteract it, a Cuckoo rule may be applied~\cite{awerbuchDHT}. If a node joins a shard, another arbitrary node has to leave this shard and join another shard. Thus, the adversary loses control over fault node shard selection. 

\ \\
\textbf{Implementation considerations.} As with any blockchain, to improve performance, multiple transactions may be confirmed in a block in \emph{SmartShards}. Thus, a single round of agreement is required for confirmation of all transactions in the block. In this case, the peers of the shard confirm several transactions as a block and then carry out the transfer part of each transaction individually. The block may contain regular UTXO transactions as well as peer join and leave transactions. 

In \emph{SmartShards}, peers may belong to more than two shards. This simplifies inter-shard communication and may improve the system overall performance. However, the performance gain may be limited since each peer has to participate in consensus of every shard that it belongs to.

Let us comment on the implementation of the Shard Membership Service (SMS) component of \emph{SmartShards}, which clients use to keep track of membership changes. Effectively, the SMS is a trusted source of membership information. As such, it may be implemented as a side-chain of membership records, a bootstrap service, or a Byzantine-robust host that reliably answers client membership queries.

\vspace{-0.4cm}
\section{Concluding Remarks}

\vspace{-.2cm}
In this paper we explored the idea of overlapping shards. We showed how it can be used to improve efficiency and robustness of a sharded blockchain. In the future, it would be interesting to implement \emph{SmartShards} in a real pear-to-peer system and compare it against traditional non-overlapping sharding blockchains. Alternatively, \emph{SmartShards} may be used directly in existing blockchains to overlap shards and enhance their fault-tolerance and functional performance.

\newpage

\bibliographystyle{plain}
\bibliography{smartShards}
\end{document}

%% file: variables.tex
\begin{algorithm}[H]
\scriptsize
\SetKwComment{Comment}{$\triangleright$\ }{}
\SetCommentSty{cmr}

\textbf{Constants}\\
$p$ \Comment*[f]{peer id}\\
$Shards$ \Comment*[f]{set of all shard ids}

$shard_1, shard_2 \in Shards$ \Comment*[f]{ids of the shards this peer is in}\\


$Peers_1$, $Peers_2$ \Comment*[f]{ sets of peer ids for $shard_1$ and $shard_2$}\\

$CounterShards_1 \equiv \{ \langle q,s \rangle \ \ | \ \ q \in Peers_1 , s \in Shards \}$  \Comment*[f]{
    countershards for $Peers_1$
}\\

$CounterShards_2 \equiv \{ \langle q,s \rangle \ \ | \ \ q \in Peers_2 , s \in Shards \}$  \Comment*[f]{
    countershards for $Peers_1$
}\\

\BlankLine \BlankLine
\textbf{Transaction variables}\\
$client$ \Comment*[f]{client initiating the transaction}\\
$\mathit{sWallet}$ \Comment*[f]{source wallet}\\
$\mathit{tWallet}$ \Comment*[f]{target wallet}\\
$\mathit{UTXO}$ \Comment*[f]{source unspent transaction output}\\
$sShard$ \Comment*[f]{shard containing the source wallet}\\
$tShard$ \Comment*[f]{shard containing the target wallet}\\
\BlankLine \BlankLine
\textbf{Process variables}\\
$Ledger_1$, $Ledger_2$ \Comment*[f]{sequence of committed transactions in $shard_1$, $shard_2$}\\
\Comment*[f]{record format: $\langle \mathit{UTXO},\ sShard, tShard, \mathit{sWallet}, \mathit{tWallet} \rangle$}\\

$\mathit{Transfers}_1, \mathit{Transfers}_2$ \Comment*[f]{set of cross-shard transaction receipts for $shard_1$, $shard_2$}\\
\Comment*[f]{record format: $\langle \mathit{UTXO},\ sShard, tShard, \mathit{sWallet}, \mathit{tWallet}, q \rangle$}\\
\label{alg:EndOfVars}

\caption{$SmartShards$, Variables}
\label{algVars}
\end{algorithm}

%% file: algorithm.tex
\begin{algorithm}[H]
\scriptsize
\SetKwComment{Comment}{$\triangleright$\ }{}
\SetCommentSty{cmr}
\SetFuncSty{textsc}
\SetKwBlock{function}{function}{:}{}
\SetKwBlock{Upon}{upon}{:}{}
\SetKwBlock{Upon}{upon}{}

\function( \textsc{Valid}$\langle client, \mathit{UTXO}, sShard \rangle$) {
    $\triangleright$ returns \textbf{true} if $\mathit{UTXO}$ is owned by $client$ and is unspent
}

\BlankLine

\BlankLine
$\triangleright$ actions for $shard_1$\\
\Upon(\label{transaction} \textbf{receive} \textsc{Transaction} $\langle client, \mathit{UTXO}, sShard, tShard, \mathit{sWallet}, \mathit{tWallet} \rangle$) {
    \If{$\textsc{Valid}(client, \mathit{UTXO}, sShard)$ \textnormal{and} $sShard = shard_1$} {
        $\textsc{StartConsensusTx}(\mathit{UTXO}, sShard, tShard, \mathit{sWallet}, \mathit{tWallet})$\\
        
    }
}

\BlankLine

\Upon(\label{endConsensusTx} \textsc{EndConsensusTx}$\langle \mathit{UTXO}, sShard, tShard, \mathit{sWallet}, \mathit{tWallet} \rangle$){
    \If{
    $\textsc{Valid}(client, \mathit{UTXO}, sShard)$} {
        \textbf{append} $\langle \mathit{UTXO}, sShard, tShard, \mathit{sWallet}, \mathit{tWallet}  \rangle$ to $Ledger_1$\\
        \eIf{$sShard \neq tShard$ \textnormal{and} $tShard$ \textnormal{=} $shard_2$} {
            \textbf{send} \textsc{Transfer} ($\mathit{UTXO}, sShard, tShard, \mathit{sWallet}, \mathit{tWallet}$) to $tShard$\\
        } {
        \textbf{send} \textsc{TxConfirmation} $\langle p, \mathit{UTXO} \rangle$ to $client$
        }
    }
}

\Upon(\label{endConsensusTr} \textsc{EndConsensusTr}$\langle \mathit{UTXO}, sShard, tShard, \mathit{sWallet}, \mathit{tWallet}, result \rangle$){
    \If{$result$} {
        \textbf{append} $\langle \mathit{UTXO}, sShard, tShard, \mathit{sWallet}, \mathit{tWallet}  \rangle$ to $Ledger_1$\\
        \textbf{send} \textsc{TxConfirmation} $\langle p, \mathit{UTXO} \rangle$ to $client$
    }
}

\BlankLine

\Upon(\label{transfer} \textbf{receive} \textsc{Transfer}$\langle \mathit{UTXO}, fromShard, toShard, \mathit{sWallet}, \mathit{tWallet} \rangle$ from $q$) {
    \If{$q \in Peers_1$ \textnormal{and} $CounterShards_1.q$ \textnormal{=} $fromShard$ \textnormal{and}\\
    \ \ \ \
    $\langle \mathit{UTXO}, fromShard, toShard, \mathit{sWallet}, \mathit{tWallet}, q \rangle$ $\not \in$ $\mathit{Transfers}_1$
    \textnormal{and}\\
    \ \ \ \
    $toShard$ \textnormal{=} $shard_1$}
    {
        \textbf{add} $\langle \mathit{UTXO}, fromShard, toShard, \mathit{sWallet}, \mathit{tWallet}, q \rangle$ to $\mathit{Transfers}_1$\\
        \If{ $\# \{ \langle \mathit{UTXO}, fromShard, toShard, \mathit{sWallet}, \mathit{tWallet}, * \rangle$  $\in$ $\mathit{Transfers}_1 \} \geq \lceil  \# 
        \{ \langle q,s \rangle \in CounterShards_1 : s = sShard \} / 2 \rceil$ $\textnormal{and}$ $\langle \mathit{UTXO}, fromShard, toShard, \mathit{sWallet}, \mathit{tWallet} \rangle$ $\not \in$ $\mathit{Ledger}_1$ }{
             $\textsc{StartConsensusTr}(\mathit{UTXO}, fromShard, toShard, \mathit{sWallet}, \mathit{tWallet}, \textbf{true})$\\
        }
    }
}

\BlankLine
$\triangleright$ actions for $shard_2$ similar\\
\caption{$SmartShards$, Actions}
\label{algActions}
\end{algorithm}

%% file: algorithm3.tex
\begin{algorithm}[H]
\scriptsize
\SetKwComment{Comment}{$\triangleright$\ }{}
\SetCommentSty{cmr}
\SetFuncSty{textsc}
\SetKwBlock{function}{function}{:}{}
\SetKwBlock{Upon}{upon}{:}{}
\SetKwBlock{Upon}{upon}{}

\textbf{Constants}\\
$client$ \Comment*[f]{client id}\\
$\mathit{UTXO}$  \\
$f$  \Comment*[f]{consensus tolerance threshold}\\

\BlankLine \BlankLine
\textbf{Variables}\\
$confirmed$ \Comment*[f]{status of the pending transaction}\\
$SMSmembers$\Comment*[f]{set of sets of peers, history of shard membership}\\
$Confirmations$\Comment*[f]{set of peers that confirmed the transaction}\\

\BlankLine \BlankLine
\textbf{Actions}\\

\Upon(\textsc{SubmitTx}$\langle \mathit{UTXO} \rangle$) {
    $confirmed = \textbf{false}$ \\
    \While (\Comment*[f]{periodically resend transaction request}) {$\neg confirmed$ }  {
        
        \textbf{send} \textsc{Transaction}$\langle client, \mathit{UTXO} \rangle$ to $SMS$
        
    }

}

\BlankLine
\BlankLine

\Upon(\label{MembershipReceipt} \textbf{receive} \textsc{Membership} $\langle Members \rangle$ from $SMS$)  {
        $SMSmembers := \{Members\} \cup SMSmembers$\\
}

\BlankLine

\Upon(\label{TxConfirmation} \textbf{receive} \textsc{TxConfirmation} $\langle peer, \mathit{UTXO} \rangle$ from $peer$)  {
    
    $Confirmations := {Confirmations} \cup peer$\\
        \If {$\exists$ $Members \in SMSmembers: \# \{Members \cap Confirmations\} > f $} {
            $confirmed := \textbf{true}$
        }
        
}

\BlankLine

$\triangleright$ join and leave actions similar

\caption{$SmartShards$, Client Actions and Variables}
\label{algClient}
\end{algorithm}

%% file: main.bbl
\begin{thebibliography}{10}

\bibitem{SmartShardsgithub}
Smartshards implementation.
\newblock \url{https://github.com/QuantasSupport/Quantas/tree/SmartShards}.

\bibitem{smartShardsData}
Smartshards perfromance evaluation data.
\newblock \url{http://www.cs.kent.edu/~mikhail/Research/trail.output.tar.gz},
  February 2025.

\bibitem{agbo2019blockchain}
Cornelius~C Agbo, Qusay~H Mahmoud, and J~Mikael Eklund.
\newblock Blockchain technology in healthcare: a systematic review.
\newblock In {\em Healthcare}, volume~7, page~56. MDPI, 2019.

\bibitem{andoni2019blockchain}
Merlinda Andoni, Valentin Robu, David Flynn, Simone Abram, Dale Geach, David
  Jenkins, Peter McCallum, and Andrew Peacock.
\newblock Blockchain technology in the energy sector: A systematic review of
  challenges and opportunities.
\newblock {\em Renewable and sustainable energy reviews}, 100:143--174, 2019.

\bibitem{androulaki2018hyperledger}
Elli Androulaki, Artem Barger, Vita Bortnikov, Christian Cachin, Konstantinos
  Christidis, Angelo De~Caro, David Enyeart, Christopher Ferris, Gennady
  Laventman, Yacov Manevich, et~al.
\newblock Hyperledger fabric: a distributed operating system for permissioned
  blockchains.
\newblock In {\em Proceedings of the thirteenth EuroSys conference}, pages
  1--15, 2018.

\bibitem{awerbuchDHT}
Baruch Awerbuch and Christian Scheideler.
\newblock Towards a scalable and robust dht.
\newblock In {\em Proceedings of the eighteenth annual ACM symposium on
  Parallelism in algorithms and architectures}, pages 318--327, 2006.

\bibitem{baldoni2013protocol}
Roberto Baldoni, Silvia Bonomi, and Amir~Soltani Nezhad.
\newblock A protocol for implementing byzantine storage in churn-prone
  distributed systems.
\newblock {\em Theoretical Computer Science}, 512:28--40, 2013.

\bibitem{bessani2020byzantine}
Alysson Bessani, Eduardo Alchieri, Jo{\~a}o Sousa, Andr{\'e} Oliveira, and
  Fernando Pedone.
\newblock From byzantine replication to blockchain: Consensus is only the
  beginning.
\newblock In {\em 2020 50th Annual IEEE/IFIP International Conference on
  Dependable Systems and Networks (DSN)}, pages 424--436. IEEE, 2020.

\bibitem{pbft}
Miguel Castro and Barbara Liskov.
\newblock Practical byzantine fault tolerance and proactive recovery.
\newblock {\em ACM Trans. Comput. Syst.}, 20(4):398--461, November 2002.

\bibitem{dujak2019blockchain}
Davor Dujak and Domagoj Sajter.
\newblock Blockchain applications in supply chain.
\newblock {\em SMART supply network}, pages 21--46, 2019.

\bibitem{foreback2016infinite}
Dianne Foreback, Mikhail Nesterenko, and S{\'e}bastien Tixeuil.
\newblock Infinite unlimited churn.
\newblock {\em arXiv preprint arXiv:1608.00726}, 2016.

\bibitem{foreback2019churn}
Dianne Foreback, Mikhail Nesterenko, and S{\'e}bastien Tixeuil.
\newblock Churn possibilities and impossibilities.
\newblock In {\em Networked Systems: 6th International Conference, NETYS 2018,
  Essaouira, Morocco, May 9--11, 2018, Revised Selected Papers 6}, pages
  303--317. Springer, 2019.

\bibitem{garcia2019federated}
{\'A}lvaro Garc{\'\i}a-P{\'e}rez and Alexey Gotsman.
\newblock Federated byzantine quorum systems.
\newblock In {\em 22nd International Conference on Principles of Distributed
  Systems (OPODIS 2018)}. Schloss-Dagstuhl-Leibniz Zentrum f{\"u}r Informatik,
  2019.

\bibitem{hafid2020scaling}
Abdelatif Hafid, Abdelhakim~Senhaji Hafid, and Mustapha Samih.
\newblock Scaling blockchains: A comprehensive survey.
\newblock {\em IEEE access}, 8:125244--125262, 2020.

\bibitem{ito2019critical}
Kensuke Ito and Marcus O’Dair.
\newblock A critical examination of the application of blockchain technology to
  intellectual property management.
\newblock {\em Business Transformation through Blockchain: Volume II}, pages
  317--335, 2019.

\bibitem{kokoris2018omniledger}
Eleftherios Kokoris-Kogias, Philipp Jovanovic, Linus Gasser, Nicolas Gailly,
  Ewa Syta, and Bryan Ford.
\newblock Omniledger: A secure, scale-out, decentralized ledger via sharding.
\newblock In {\em 2018 IEEE symposium on security and privacy (SP)}, pages
  583--598. IEEE, 2018.

\bibitem{kuhn2010towards}
Fabian Kuhn, Stefan Schmid, and Roger Wattenhofer.
\newblock Towards worst-case churn resistant peer-to-peer systems.
\newblock {\em Distributed Computing}, 22:249--267, 2010.

\bibitem{lafourcade2020security}
Pascal Lafourcade, Mike Nopere, J{\'e}r{\'e}my Picot, Daniela Pizzuti, and
  Etienne Roudeix.
\newblock Security analysis of auctionity: a blockchain based e-auction.
\newblock In {\em Foundations and Practice of Security: 12th International
  Symposium, FPS 2019, Toulouse, France, November 5--7, 2019, Revised Selected
  Papers 12}, pages 290--307. Springer, 2020.

\bibitem{lamport1982byzantine}
Leslie Lamport, Robert Shostak, and Marshall Pease.
\newblock The byzantine generals problem.
\newblock {\em ACM Transactions on Programming Languages and Systems},
  4(3):382--401, 1982.

\bibitem{liu2023survey}
Xinmeng Liu, Haomeng Xie, Zheng Yan, and Xueqin Liang.
\newblock A survey on blockchain sharding.
\newblock {\em ISA transactions}, 141:30--43, 2023.

\bibitem{luu2016secureelastico}
Loi Luu, Viswesh Narayanan, Chaodong Zheng, Kunal Baweja, Seth Gilbert, and
  Prateek Saxena.
\newblock A secure sharding protocol for open blockchains.
\newblock In {\em Proceedings of the 2016 ACM SIGSAC conference on computer and
  communications security}, pages 17--30, 2016.

\bibitem{malkhi1998byzantine}
Dahlia Malkhi and Michael Reiter.
\newblock Byzantine quorum systems.
\newblock {\em Distributed computing}, 11(4):203--213, 1998.

\bibitem{bitcoin}
Satoshi Nakamoto.
\newblock Bitcoin: A peer-to-peer electronic cash system.
\newblock 2008.

\bibitem{quantas}
Joseph Oglio, Kendric Hood, Mikhail Nesterenko, and Sebastien Tixeuil.
\newblock Quantas: quantitative user-friendly adaptable networked things
  abstract simulator.
\newblock In {\em Proceedings of the 2022 Workshop on Advanced tools,
  programming languages, and PLatforms for Implementing and Evaluating
  algorithms for Distributed systems}, pages 40--46, 2022.

\bibitem{samaniego2016blockchain}
Mayra Samaniego, Uurtsaikh Jamsrandorj, and Ralph Deters.
\newblock Blockchain as a service for iot.
\newblock In {\em 2016 IEEE international conference on internet of things
  (iThings) and IEEE green computing and communications (GreenCom) and IEEE
  cyber, physical and social computing (CPSCom) and IEEE smart data
  (SmartData)}, pages 433--436. IEEE, 2016.

\bibitem{wang2019monoxide}
Jiaping Wang and Hao Wang.
\newblock Monoxide: Scale out blockchains with asynchronous consensus zones.
\newblock In {\em 16th USENIX symposium on networked systems design and
  implementation (NSDI 19)}, pages 95--112, 2019.

\bibitem{wendl2023environmental}
Moritz Wendl, My~Hanh Doan, and Remmer Sassen.
\newblock The environmental impact of cryptocurrencies using proof of work and
  proof of stake consensus algorithms: A systematic review.
\newblock {\em Journal of Environmental Management}, 326:116530, 2023.

\bibitem{white2022characterizing}
Bryan White, Aniket Mahanti, and Kalpdrum Passi.
\newblock Characterizing the opensea nft marketplace.
\newblock In {\em Companion Proceedings of the Web Conference 2022}, pages
  488--496, 2022.

\bibitem{ethereum}
Gavin Wood.
\newblock Ethereum: A secure decentralized generalized transaction ledger.
\newblock {\em Ethereum project yellow paper}, 151:1--32, 2014.

\bibitem{yu2020survey}
Guangsheng Yu, Xu~Wang, Kan Yu, Wei Ni, J~Andrew Zhang, and Ren~Ping Liu.
\newblock Survey: Sharding in blockchains.
\newblock {\em IEEE Access}, 8:14155--14181, 2020.

\bibitem{zamani2018rapidchain}
Mahdi Zamani, Mahnush Movahedi, and Mariana Raykova.
\newblock Rapidchain: Scaling blockchain via full sharding.
\newblock In {\em Proceedings of the 2018 ACM SIGSAC conference on computer and
  communications security}, pages 931--948, 2018.

\end{thebibliography}
